\documentclass{llncs}
\usepackage[utf8]{inputenc}

\usepackage{amsmath,amssymb,amsthm}
\usepackage{float}
\usepackage{graphicx}
\usepackage[vlined,ruled,linesnumbered]{algorithm2e}
\usepackage{url}

\usepackage{times}
\usepackage{xspace}

\newcommand{\A}{\mathcal{A}} 
 
\newcommand{\E}{\mathcal{E}} 
 
\newcommand{\I}{\mathcal{I}}

\renewcommand{\O}{\mathcal{O}} \renewcommand{\P}{\mathcal{P}}
 
\renewcommand{\S}{\mathcal{S}} \newcommand{\T}{\mathcal{T}}

 \newcommand{\X}{\mathcal{X}}
\newcommand{\Y}{\mathcal{Y}} \newcommand{\Z}{\mathcal{Z}}

\newcommand{\LTL}{{\sc ltl}\xspace}

\newcommand{\NNF}{{\sc nnf}\xspace}

\newcommand{\Nat}{{\rm I\kern-.23em N}}

\def\tt{\top}
\def\ff{\bot}
\def\X{\mathcal{X}}
\def\Y{\mathcal{Y}}

\begin{document}

\title{A Symbolic Approach to Safety \LTL Synthesis}

\author{%
Shufang Zhu\inst{1}  \and 
Lucas M. Tabajara\inst{2} \and
Jianwen Li\inst{2}\thanks{Corresponding author} \and 
Geguang Pu\inst{1}\thanks{Corresponding author} \and 
Moshe Y. Vardi\inst{2}
}%

\institute{East China Normal University, Shanghai, China
\and
Rice University, Texas, USA}

\maketitle

\begin{abstract}
Temporal synthesis is the automated design of a system that interacts with an environment, using the declarative specification of the system's behavior. A popular language for providing such a specification is Linear Temporal Logic, or \LTL. \LTL synthesis in the general case has remained, however, a hard problem to solve in practice. Because of this, many works have focused on developing synthesis procedures for specific fragments of \LTL, with an easier synthesis problem. In this work, we focus on Safety \LTL, defined here to be the Until-free fragment of \LTL in Negation Normal Form~(\NNF), and shown to express a fragment of safe \LTL formulas. The intrinsic motivation for this fragment is the observation that in many cases it is not enough to say that something ``good'' will eventually happen, we need to say by when it will happen. We show here that Safety \LTL synthesis is significantly simpler algorithmically than \LTL synthesis.
We exploit this simplicity in two ways, first by describing an explicit approach based on a reduction to Horn-SAT, which can be solved in linear time in the size of the game graph, and then through an efficient symbolic construction, allowing a BDD-based symbolic approach which significantly outperforms extant \LTL-synthesis tools.
\end{abstract}

\section{Introduction}\label{sec:intro}

Research on synthesis is the culmination of the ideal of declarative programming. By describing a system in terms of what it should do, rather than how it should be done, we are able to simplify the design process while also avoiding human mistakes. In the framework defined by synthesis, we describe a specification of a system's behavior in a formal language, and the synthesis procedure automatically designs a system satisfying this specification~\cite{MannaWaldinger71}.
Reactive synthesis~\cite{PnueliR89} is one of the most popular variants of this problem, in which we wish to synthesize a system that interacts continuously with an environment. Such systems include, for example, operating systems and controllers for mechanical devices. To specify the behavior of such systems, we need a specification language that can reason about changes over time. A popular such language is Linear Temporal Logic, or \LTL \cite{Pnu77}.

Despite extensive research, however, synthesis from \LTL formulas remains a difficult problem. The classical approach is based on translating the formula to a deterministic parity automaton and reducing the synthesis problem to solving a parity game~\cite{PnueliR89}. This translation, however, is not only theoretically hard, given its doubly-exponential upper bound, but also inefficient in practice due to the lack of practical algorithms for determinization~\cite{DFogartyKVW13}. Furthermore, despite the recent quasi-polynomial algorithm~\cite{CaludeJKL017} for parity games, it is still not known if they can be solved efficiently. 
A promising approach to mitigating this problem consists of developing synthesis techniques for certain fragments of \LTL that cover interesting classes of specifications but for which the synthesis problem is easier. Possibly the most notable example is that of Generalized Reactivity(1) formulas, or GR(1)~\cite{gr1}, a fragment for which the synthesis problem can be solved in cubic time with respect to the game graph.

Here we focus on the Safety \LTL fragment, which we define to be the fragment of \LTL composed of Until-free formulas in Negation Normal Form (\NNF). Such formulas express~\emph{safety properties}, meaning that every violating trace has a finite bad prefix that falsifies the formula~\cite{KV01}.
The intrinsic motivation for this fragment is the observation that in many cases it is not enough to say that something “good” will eventually happen, we need to say by \emph{when} it will happen \cite{Lamport83}.
For this strict subset of \LTL, the synthesis problem can be reduced to a \emph{safety game}, which is far easier to solve. In fact, for such a game the solution can be computed in linear time with respect to the game graph~\cite{AlfaroHK98}. 
Some novel techniques for safety game solving have been developed in the
context of the Annual Synthesis Competition
(SyntComp)~\footnote{\url{http://www.syntcomp.org/}}, but there the input
consists of an AIGER model, while in this paper we are concerned with
synthesis from Safety \LTL formulas. See further discussion in the Concluding Remarks.

Our first contribution is a new solution to safety games by reducing to Horn satisfiability (Horn-SAT). There have been past works using SAT in the context of bounded synthesis~\cite{bloem2014sat}, 
but our approach is novel in using a reduction to Horn-SAT, which can be solved in linear time~\cite{DowlingG84}. Because, however, the Horn formula is proportional to the size of the state graph, in which the number of transitions is exponential in the number of input/output variables
and the number of states can be in the worst case doubly exponential in the size of the Safety \LTL formula, this approach becomes infeasible for larger instances. To avoid this problem, we pursue an alternative approach that uses a symbolic representation of the game graph via Binary Decision Diagrams (BDDs)~\cite{Bryant92}.

Symbolic solutions to safety games have played an important part in \LTL synthesis tools following the idea of \emph{Safraless synthesis}~\cite{KupfermanV05}, which avoids the high cost of determinization and the parity acceptance condition of classical \LTL synthesis by instead employing a translation to universal co-B\"{u}chi automata. Unbeast~\cite{Unbeast}, a symbolic BDD-based tool for bounded synthesis, decomposes the \LTL specification into safety and non-safety parts, using an incremental bound to allow the non-safety part to also be encoded as a safety game. Another tool, Acacia+~\cite{Acacia}, takes a bounded synthesis approach that allows the synthesis problem to be reduced to a safety game, then explores the structure of the resulting game to implement a symbolic antichain-based algorithm.

In the above approaches the safety game is constructed from a co-B\"{u}chi automaton of the \LTL specification. Our insight in this paper is that, since every bad trace of a formula in the Safety \LTL fragment has a finite prefix, we can construct from the negation of such a formula a deterministic finite automaton that accepts exactly the language corresponding to bad prefixes. This DFA can be seen as the dual of a \emph{safety automaton} defining a safety game over the same state space. Using a DFA as the basis for our safety game allows us to leverage tools and techniques developed for symbolic construction, determinization and minimization of finite automata.

Our symbolic synthesis framework is inspired by a recent approach~\cite{ZTLPV17} for synthesis of \LTL over finite traces.
This problem can be seen as the dual of Safety \LTL synthesis, and as such we can inter-reduce the realizability problem between the two by negating the result. Nevertheless, the strategy generation is irreducible since the two problems are 
solving the game for different players. Therefore, we modify the algorithm to produce a strategy for the safety game instead.
The procedure consists of 
two phases. First we construct symbolically a safety automaton from the Safety \LTL formula 
instead of direct construction. For that we present a translation from the negation of Safety \LTL to first-order logic over finite traces, which allows us to symbolically construct the dual DFA of the safety automaton. Second, we solve the safety game by computing the set of winning states through a backwards symbolic fixpoint computation, and then applying a boolean-synthesis procedure~\cite{DrLu.cav16} to symbolically construct a strategy.

In summary, our contribution in this paper is to introduce a fragment of \LTL called Safety \LTL and present two approaches for the synthesis problem for this fragment, an explicit one based on a reduction to Horn-SAT and a symbolic one exploiting techniques for symbolic DFA construction. 
Since Safety \LTL is a fragment of general \LTL, existing \LTL synthesis tools can likewise be used to solve the Safety \LTL synthesis problem. To demonstrate the benefits of developing specialized synthesis techniques, we perform an experimental comparison with Unbeast and Acacia+, both tools for general \LTL synthesis. 
Our results show that the explicit approach is able to outperform these tools when the formula is small, while the symbolic approach has the best performance overall.
\section{Preliminaries}\label{sec:pre}
\subsection{Safety/Co-safety \LTL}
Linear Temporal Logic (\LTL), first introduced in~\cite{Pnu77}, extends propositional logic by introducing temporal operators. Given a set $\P$ of propositions, the syntax of \LTL formulas is defined as~$\phi ::= \tt\ |\ \ff\ |\ p\ |\ \neg \phi\ |\ \phi_1\wedge\phi_2\ |\ X\phi\ |\ \phi_1 U \phi_2$.

$\tt$ and $\ff$ represent \textit{true} and \textit{false} respectively. $p\in \P$ is an 
\textit{atom}, and we define a literal $l$ to be an atom or the negation of an atom. $X$ (Next) and $U$ (Until)
are temporal operators. We also introduce the dual operator of $U$, namely $R$ (Release), defined as $\phi_1 R\phi_2\equiv \neg (\neg\phi_1 U\neg \phi_2)$. 
Additionally, we define the abbreviations $F\phi\equiv \tt U\phi$ and $G\phi\equiv \ff R\phi$. Standard boolean abbreviations, such as $\vee$ (or) and $\rightarrow$ (implies) are also used. An \LTL formula $\phi$ is \emph{Until-free}/\emph{Release-free} iff it does not contain the Until/Release operator. Moreover, we say $\phi$ is in Negation Normal Form (\NNF), iff all negation operators in $\phi$ are pushed only in front of atoms. 

A \textit{trace} $\rho = \rho_0\rho_1\ldots$ is a sequence of propositional interpretations (sets), in which 
$\rho_m\in 2^\P$ ($m \geq 0$) is the $m$-th interpretation of $\rho$, and $|\rho|$ represents the length of $\rho$. Intuitively, $\rho_m$ is interpreted as the set of propositions which are $true$ at instant $m$.
Trace $\rho$ is an \textit{infinite} trace if $|\rho| = \infty$, which is formally denoted as $\rho\in (2^\P)^{\omega}$. 
Otherwise $\rho$ is a \textit{finite} trace, denoted as $\rho\in (2^\P)^{*}$. \LTL formulas are interpreted over infinite traces. Given an infinite trace $\rho$ and an \LTL formula $\phi$, we inductively define when $\phi$ is $true$ in $\rho$ at step $i$ ($i \geq 0$), written $\rho, i \models \phi$, as follows: 
\begin{itemize}
  \item $\rho, i \models\tt$ and $\rho, i \not\models\ff$;
  \item $\rho, i \models p$ iff $p \in \rho_i$;
  \item $\rho, i \models \neg \phi$ iff $\rho,i \not\models \phi$;
  \item $\rho, i \models\phi_1 \wedge \phi_2$, iff $\rho,i \models \phi_1$ and $\rho, i \models \phi_2$;
  \item $\rho, i \models X\phi$, iff $\rho, i+1 \models \phi$;
  \item $\rho, i \models \phi_1 U \phi_2$, iff there exists  $j \geq i$ such that $\rho, j\models \phi_2$, and for all  $i \leq k < j$, we have $\rho, k \models \phi_1$.
\end{itemize}

An \LTL formula $\phi$ is $true$ in $\rho$, denoted by $\rho \models \phi$, if and only if $\rho, 0\models\phi$. 

Informally speaking, a \emph{safe} \LTL formula rejects traces whose ``badness'' follows from a finite prefix. Dually, a \emph{co-safe} \LTL formula accepts traces whose ``goodness'' follows from a finite prefix. Thus, $\phi$ is a safe formula iff $\neg\phi$ is a co-safe formula. To define the \emph{safe}/\emph{co-safe} formulas, we need to introduce the concept of \emph{bad}/\emph{good prefix}.

\begin{definition}[Bad/Good Prefix~\cite{KV01}]
Consider a language L of infinite words over $\P$. A finite word x over $\P$ is a bad/good prefix for L if and only if for all infinite words $y$ over $\P$, the concatenation $x \cdot y$ of $x$ and $y$ isn't/is in L.
\end{definition}

\noindent
Safe/co-safe \LTL formulas are defined as follows.

\begin{definition}[safe/co-safe~\cite{KV01}]\label{def:safe}
An \LTL formula $\phi$ is safe/co-safe iff every word that violates/satisfies $\phi$ has a bad/good prefix. 
\end{definition}

We use $pref(\phi)$ to denote the set of bad prefixes for safe formula $\phi$, equivalently, we denote by $co$-$pref(\neg \phi)$, the set of good prefixes for $\neg \phi$, which is co-safe. Indeed, $pref(\phi) = co$-$pref(\neg \phi)$~\cite{KV01}.

\begin{theorem}\label{thm:equiv}
An \LTL formula $\phi$ is safe iff $\neg \phi$ is co-safe, and each bad prefix for safe formula $\phi$ is a good prefix for $\neg \phi$.
\end{theorem}

Checking if a given \LTL formula is safe/co-safe is PSPACE-complete~\cite{KV01}. We now introduce a fragment of \LTL where safety/co-safety is a syntactical feature. 

\begin{theorem}[\cite{Sistla94}]
If an \LTL formula $\phi$ in \NNF is Until-free/Release-free, then $\phi$ is safe/co-safe.
\end{theorem}

Motivated by this theorem, we define now the syntactic fragment of
\emph{Safety/Co-Safety \LTL}.

\begin{definition}
Safety/Co-Safety \LTL formulas are in \NNF and Until-free/Release-free, respectively.
\end{definition}

\noindent
{\sf Remark}:
To the best of our knowledge, it is an open question whether every safe \LTL formula is equivalent to some Safety \LTL formula. We conjecture that this is the case.

\subsection{Boolean Synthesis}
In this paper, we utilize the \emph{boolean synthesis} technique proposed in~\cite{DrLu.cav16}. 


\begin{definition}[Boolean Synthesis~\cite{DrLu.cav16}]
Given two disjoint atom sets $\I, \O$ of input and output variables, respectively, and a boolean formula $\xi$ over $\I\cup \O$, the boolean-synthesis problem is to construct a function $\gamma:2^{\I}\rightarrow 2^{\O}$ such that, for all $I \in 2^\I$, if there exists $O \in 2^\O$ such that $I \cup O \models \xi$, then $I \cup \gamma(I) \models \xi$. 
We call $\gamma$ the \emph{implementation function}.
\end{definition}

We treat boolean synthesis as a black box, applying it to the key operation of Safety \LTL synthesis proposed in this paper. For more details on algorithms and techniques for boolean synthesis we refer to~\cite{DrLu.cav16}. 


\section{Safety \LTL Synthesis}\label{sec:theory}
In this section we give the definition of Safety \LTL synthesis. We then show how this problem can be modeled as a \emph{safety game} played over a kind of deterministic automaton, called a \emph{safety automaton}. In the following sections we describe approaches to construct this automaton from a Safety \LTL formula and solve the game that it specifies.



\begin{definition}[Safety \LTL Synthesis]\label{def:syn}
Let $\phi$ be an \LTL formula over an alphabet $\mathcal{P}$ and $\X, \Y$ be two disjoint atom sets such that $\X \cup \Y = \mathcal{P}$. $\X$ is the set of \emph{input (environment) variables} and $\Y$ is the set of \emph{output (controller) variables}. $\phi$ is \emph{realizable} with respect to $\langle\X, \Y\rangle$ if there exists a strategy $g: (2^{\X})^* \rightarrow 2^{\Y}$, such that for an arbitrary infinite sequence $X_0,X_1,\ldots\in (2^{\X})^{\omega}$, $\phi$ is true in the infinite trace $\rho= (X_0\cup g(X_0)), (X_1\cup g(X_0,X_1)), (X_2\cup g(X_0,X_1,X_2))\ldots$. The \emph{synthesis} procedure is to compute such a strategy if $\phi$ is \emph{realizable}.
\end{definition}

There are two versions of the Safety \LTL synthesis, depending on the first player. Here we consider that the environment moves first, but the version where the controller moves first can be obtained by a small modification.

The Safety \LTL synthesis is a subset of \LTL synthesis by restricting the property to be a Safety \LTL formula. Therefore, we can use general \LTL-synthesis methods to solve the Safety \LTL synthesis problem. Classical approaches to \LTL synthesis problems involve two steps: 1) Convert the \LTL formula to a deterministic automaton; 2) Reduce \LTL synthesis to an infinite game over the automaton. We now present the automata corresponding to the class of Safety \LTL formulas.

\begin{definition}[Deterministic Safety Automata]
A deterministic safety automaton (DSA) is a tuple
$A^s = (2^\P, S, s_0, \delta)$, where $2^{\P}$ is the alphabet, $S$ is a finite set of states
with $s_0$ as the initial state, and $\delta : S \times 2^{\P} \rightarrow S$ is a partial transition function.
Given an infinite trace $\rho\in (2^{\mathcal{P}})^{\omega}$, a run $r$ of $\rho$ on $A^s$ is a sequence of states $s_0,s_1, s_2,\ldots$ such that $s_{i+1} = \delta (s_i,\rho_i)$. $\rho$ is accepted by $A^s$ if $A^s$ has an infinite run $r$ of $\rho$. 
\end{definition}

Note that in the definition, $\delta$ is a partial function, meaning that given $s\in S$ and $a\in 2^{\mathcal{\P}}$, $\delta (s, a)$ can either return a state $s'\in S$ or be undefined. Thus, an infinite run of $\rho$ on $A^s$ may not exist due to the possibility of $\delta(s_i,\rho_i)$ being undefined for some $(s_i, \rho_i)$. 
A DSA is essentially a deterministic B\"{u}chi automaton (DBA)~\cite{Buchi} with a partial transition function and a set of accepting states $F = S$.

\emph{Deterministic safety games} are games between two players, the environment and the controller, played over a DSA. We have two disjoint sets of variables $\mathcal{X}$ and $\mathcal{Y}$. $\mathcal{X}$ contains uncontrollable variables, which are under the control of the environment. $\mathcal{Y}$ contains controllable variables, which are under the control of the controller. A \emph{round} consists of both the controller and the environment setting the value of the variables they control. A \emph{play} of the game is a word $\rho \in (2^{\mathcal{X} \cup \mathcal{Y}})^\omega$ that describes how the environment and the controller set values to the variables during each round. A \emph{run} of the game is the corresponding sequence of states through the play.
The \emph{specification} of the game is given by a deterministic safety automaton $A^s = (2^{\mathcal{X} \cup \mathcal{Y}}, S, s_0, \delta)$.

A \emph{winning} play for the controller is an infinite sequence 
accepted by $A^s$. A \emph{strategy} for the controller is a function $f : (2^{\mathcal{X}})^* \rightarrow 2^{\mathcal{Y}}$ such that given a history of the setting of the environmental variables, $f$ determines how the controller set the controllable variables in $\mathcal{Y}$. A strategy is a \emph{winning strategy} if starting from the initial state $s_0$, for every possible sequence of assignments of the variables in $\X$, it leads to an infinite run. Checking the existence of such a \emph{winning strategy} counts for the \emph{realizability} problem. 

Safety games can be seen as duals of reachability games, where reachability games are won by reaching a set of winning states, while safety games are won by avoiding a set of losing states. Safety games however cannot be reduced to reachability games. The realizability problem of safety game can indeed be reduced to that of reachability game since the two are dual and the underlying game is determined, but this does not work for strategy generation. Safety game does not generate a winning strategy for the environment if it is unrealizable.
It is known that reachability games can be solved in linear time in the size of the game graph~\cite{AlfaroHK98}. One of the ways to do this is by a reduction to Horn Satisfiability, which can be solved in linear time~\cite{DowlingG84}. In the next section we present such a reduction.

\section{Explicit Approach to Safety Synthesis}\label{sec:horn_sat}

We now show how to solve safety games by reducing to Horn satisfiability 
(Horn-SAT), a variant of SAT where every clause has at most one
positive literal. Horn-SAT is known to be solvable in linear time using 
constraint propagation, cf.~\cite{DowlingG84}. Modern SAT solvers use
specialized data structures for performing very fast constraint propagation \cite{MZ09}. 

From a DSA $A^s = (2^{\X \cup \Y}, S, s_0,\delta)$ defining a safety game, 
we construct a Horn formula $f$ such that the game is winning for the system
if and only if $f$ is satisfiable. Then, from a satisfying assignment of $f$
we can extract a winning strategy. We now describe the construction of the
Horn formula.
There are three kinds of Boolean variables in $f$: (1) state variables: $p_s$ for each state $s \in S$; (2) state-input variables: $p_{(s,X)}$ for each state $s \in S$ and $X \in 2^{\X}$; (3) state-input-output variables: $p_{(s,X,Y)}$ for each state $s \in S$, $X \in 2^{\X}$, and $Y \in 2^{\Y}$.

We first construct a non-Horn boolean formula $f'$, then we show how to obtain a Horn formula $f$ from $f'$.
The intuition of the construction is that first, $s_0$ must be a winning 
state. Then, for every winning state, for all inputs there should exist 
an output such that the corresponding successor is a winning state.

Let $n$ represent the number of possible output 
assignments: $2^\Y=\{Y_1,\ldots,Y_n\}$, $n = 2^{|\Y|}$.
$f'$ is a conjunction of $p_{s_0}$ with the following constraints 
for each state $s \in S$:
(1) $p_{s} \rightarrow p_{(s,X)}$,
for each $X\in 2^\X$ \label{item:states};
(2) $p_{(s,X)} \rightarrow 
\left(p_{(s,X,Y_1)} \vee p_{(s,X,Y_2)} \vee \ldots \vee p_{(s,X,Y_n)}\right)$,
for each $X\in 2^\X$ \label{item:outputs};
(3) $p_{(s,X,Y)} \rightarrow p_{\delta(s,X,Y)}$, 
for each $X \in 2^{\X}$, $Y \in 2^{\Y}$,
if  $\delta(s,X,Y)$ is well defined \label{item:defined};
and (4) $\neg p_{(s,X,Y)}$,
for each $X \in 2^{\X}$, $Y \in 2^{\Y}$,
if  $\delta(s,X,Y)$ is undefined. \label{item:undefined}

\begin{theorem}
The formula $f'$ is satisfiable with assignment $\alpha'$ iff the safety game
over $A^s$ is realizable and $\alpha'$ encodes a winning strategy.
\end{theorem}

\begin{proof}
If $f'$ is satisfiable with assignment $\alpha'$, there is a set
$C \subseteq S$ of states, where for each state $s \in C$, it is the case 
that $p_s$ is true in $\alpha'$. Then, by clauses of type (1), given a state
$s \in C$, for all inputs $X \in 2^\X$, it is the case that $p_{(s,X)}$ is
also true in $\alpha'$. Furthermore, by clause of type (2), there must be 
some output $Y\in 2^\Y$ such that $p_{(s,X,Y)}$ is true in $\alpha'$.
Since $p_{(s,X,Y)}$ is true, there cannot be a clause $\neg p_{(s,X,Y)}$ of
type (4), and therefore it is the case that $\delta(s,X,Y)$ is well defined
and, by clause of type (3), $p_{\delta(s,X,Y)}$ is also true in $\alpha'$.
This means that we have a wining strategy such that all states in
$C$, including $s_0$, are winning. In response to input
$X\in 2^\X$, the system outputs $Y\in 2^\Y$
such that $p_{(s,X,Y)}$ is true in $\alpha'$, and this ensures that
the successor state $\delta(s,X,Y)$ is also in $C$.

If the safety game over $A^s$ is realizable, then there is a 
winning strategy $g:S\times 2^\X\rightarrow 2^\Y$ and a set 
$C \subseteq S$, containing $s_0$, of winning states such that for each 
state $s\in C$ and input $X\in 2^\X$, the output $Y=g(s,X)$ is such that $\delta(s,X,Y)\in C$. Then the truth assignment $\alpha'$
that makes $p_s$ true iff $s\in C$, and makes $p_{(s,X)}$ and
$p_{(s,X,g(s,X))}$ true for all $s\in C$ and $X\in 2^\X$
is a satisfying assignment of $f$.
\end{proof}

We now transform the formula $f'$ to an equi-satisfiable formula $f$ 
that is a Horn formula (in which every clause contains at most one
positive literal).
We replace each variable $p_s$, $p_{(s,X)}$, and $p_{(s,X,Y)}$ in $f'$ by
its negative literal $\neg p_s$, $\neg p_{(s,X)}$, and 
$\neg p_{(s,X,Y)}$, respectively. We can then rewrite each
constraint $(\neg p_s \rightarrow \neg p_{(s,X)})$ as $(p_{(s,X)}\rightarrow p_s)$. Similarly, we can rewrite $(\neg p_{(s,X)} \rightarrow 
(\neg p_{(s,X,Y_1)} \vee \ldots \vee \neg p_{(s,X,Y_n)}))$ as the equivalent constraint
$((p_{(s,X,Y_1)} \wedge \ldots \wedge p_{(s,X,Y_n)})\rightarrow p_{(s,X)})$.
$f$ is equivalent to $f'$ with the polarity of the literals flipped, therefore we have that $f$ is equi-satisfiable to $f'$. Given a satisfying assignment
$\alpha$ for $f$, we obtain a satisfying assignment $\alpha'$ for
$f'$ by, for every variable $p$, assigning $p$ to be true 
in $\alpha'$ iff $p$ is assigned false in $\alpha$. 

Since $f$ is a Horn formula, we can obtain a winning strategy in linear
time. Note, however, that $f$ is constructed from an explicit
representation of the DSA $A^s$, as a state graph with one transition per
assignment of the input and output variables. The challenge for this approach is the blow-up in the size of the state graph with respect to the 
input temporal formula. To address this challenge we need to be able to express the state graph more succinctly.

Therefore, we present an alternative approach for solving safety games using a symbolic representation of the state graph. Although the algorithm is no longer linear, not having to use an explicit representation of the game makes up for that fact. In order to construct this symbolic representation efficiently, we exploit the fact that safety games are dual to reachability games played over a DFA, allowing us to use techniques for symbolic construction of DFAs. This construction is described in the next section.
\section{Symbolic Approach to Safety Synthesis}

In order to perform Safety-\LTL synthesis symbolically, the first step is to construct a symbolic representation of the DSA from the Safety-\LTL formula. The following section explains how we can achieve this. The key insight that we use is that a symbolic representation of the DSA can be derived from the symbolic representation of the DFA encoding the set of bad prefixes of the Safety-\LTL formula, allowing us to exploit techniques for symbolic DFA construction. After this, we describe how we can, from this representation, symbolically compute the set of winning states of the safety game, and then extract from them a winning strategy using boolean synthesis.

\subsection{From Safety \LTL to Deterministic Safety Automata} \label{sec:construction}
In this section, we propose a conversion from Safety \LTL to DSA. The standard approach to constructing deterministic automata for \LTL formulas is to first convert an \LTL formula to a nondeterministic B\"uchi automaton using tools such as SPOT~\cite{spot}, LTL2BA~\cite{ltl2ba}, and then apply a determinization construction, e.g., Safra's construction~\cite{Safra88}.
The conversion from \LTL to deterministic automata, however, is intractable in practice, not only because of the doubly-exponential complexity, but also the non-trivial construction of both Safra~\cite{Safra88} and Safraless~\cite{KupfermanV05} approaches. Therefore, \LTL synthesis is able to benefit from a better automata construction technique. One of the contribution in this paper is proving such a technique which efficiently constructs the corresponding safety automata of Safety \LTL formulas.
The novelty here is a much simpler conversion, thus yielding a more efficient synthesis procedure.


Since every trace rejected by a DSA $A^s$ can be rejected in a finite number of steps, we can alternatively define the language accepted by $A^s$ by the finite prefixes that it rejects. This allows us to work in the domain of finite words, which can be recognized much more easily, using deterministic finite automata. Therefore, a DSA can be seen as the dual of a DFA over the same state space. Given a DFA $D = (2^{\P}, S_d, s_0, \lambda, F_d)$, the corresponding DSA $A^s = (2^{\P}, S, s_0, \delta)$ can be generated by following steps: 1) $S = S_d \backslash F_d$; 2) For $s \in S, a \in 2^{\P}$, if $\lambda(s,a) = s' \in S$, then $\delta(s,a) = s'$, otherwise $\delta(s, a)$ is undefined.

\begin{theorem}[\cite{KV01}]\label{thm:safe}
Given a Safety \LTL formula $\phi$, there is a DFA $A_{\phi}$ which accepts exactly the finite traces that are bad prefixes for $\phi$.
\end{theorem}

Given a Safety \LTL formula $\phi$ and the corresponding DFA $A_{\phi}$, we can construct the DSA $A^s_{\phi}$. The correctness of such construction is guaranteed by the following theorem.
\begin{theorem}\label{thm:dsa}
For a Safety \LTL formula $\phi$, the DSA $A^s_{\phi} = (2^{\P}, S, s_0, \delta)$, which is dual to $A_{\phi} = (2^{\P}, S_d, s_0, \lambda, F_d)$, accepts exactly the traces that satisfy $\phi$.
\end{theorem}

\begin{proof}
For an infinite trace $\rho$, $\rho \models \phi$ implies that an arbitrary prefix $\rho'$ of $\rho$ is not a bad prefix for $\phi$, so $\rho'$ cannot be accepted by $A_{\phi}$. Therefore, starting from the initial state $s_0$, $\lambda$ always returns some successor $s' \notin F_d$, so the corresponding transition is also in $A^s_{\phi}$. The run $r$ of $\rho$ on $A^s_{\phi}$ is indeed infinite. As a result, $\rho \models \phi$ implies that $\rho$ can be accepted by $A^s_{\phi}$.

On the other hand, an infinite trace $\rho$ being accepted by $A^s_{\phi}$ implies that the run $r$ of $\rho$ on $A^s_{\phi}$ is infinite. Therefore, starting from the initial state $s_0$, partial function $\delta$ can always return some successor $s' \in S$, for which $s' \notin F_d$. There is a corresponding transition in $A_{\phi}$ for each transition in $A^s_{\phi}$, then an arbitrary prefix $\rho'$ of $\rho$ is indeed can not be accepted by $A_{\phi}$, such that $\rho'$ is not a bad prefix. As a result, $\rho$ can be accepted by $A^s_{\phi}$ implies that $\rho \models \phi$.
\end{proof}

Based on Theorem~\ref{thm:dsa}, the construction of the DSA relies on the construction of the DFA for the Safety formula $\phi$. Therefore, we can leverage the techniques and tools developed for DFA construction. Although it still cannot avoid the doubly-exponential complexity, DFA construction is much simpler than that of $\omega$-automata (e.g. parity~\cite{Safra88}, or co-B\"uchi~\cite{KupfermanV05}). 
Consider a Safety \LTL formula $\phi$. From Theorem~\ref{thm:equiv} and~\ref{thm:safe}, we know that $\neg \phi$, which is co-safe, can be interpreted over finite words. Thus, we can construct the DFA $A_{\phi}$ from $\neg \phi$.
~\\
\noindent\textbf{DFA construction}
Summarily, the DFA construction is processed as follows: Given a Safety \LTL formula $\phi$, we first negate it to obtain a Co-Safety \LTL formula $\neg\phi$. Taking the translation described below, which restricts the interpretation of $\neg \phi$ over \emph{finite linear ordered traces}, we can obtain a first-order logic formula $fol()$. The DFA for such $fol()$ is obviously able to accept exactly the set of bad prefixes for $\phi$ (or say, good prefixes for $\neg\phi$).

Consider an infinite trace $\sigma = \rho_0\rho_1\dotsb\rho_n\tt\tt\dotsb$ that satisfies the Co-Safety \LTL formula $\psi = \neg \phi$ in \NNF, where the finite prefix $\rho = \rho_0\rho_1\dotsb\rho_n$ of $\sigma$ is a good prefix for $\psi$. The corresponding FOL interpretation $\mathcal{I} = (\Delta^I, \cdot^{\mathcal{I}})$ of $\rho$ is defined as follows: $\Delta^I = \{0, 1, 2,\dotsb, last\}$, where $last = |\rho| - 1$. For each $p \in \mathcal{P}$, its interpretation $p^{\mathcal{I}} = \{i~|~p \in \rho(i)\}$. Intuitively, $p^{\mathcal{I}}$ is interpreted as the set of positions where $p$ is true in $\rho$.
Then we can generate a corresponding FOL formula that opens in $x$ by a function $fol(\psi, x)$ from the Co-Safety \LTL formula and a variable $x$ where $0 \leq x \leq last$, which is defined as follows:
\begin{itemize}
\item $fol(p, x) = p(x)$ and $fol(\neg p, x) = \neg p(x)$
\item $fol(\psi_1 \wedge \psi_2, x) = fol(\psi_1, x) \wedge fol(\psi_2, x)$
\item $fol(\psi_1 \vee \psi_2, x) = fol(\psi_1, x) \vee fol(\psi_2, x)$
\item $fol(X \psi, x) = \exists y.succ(x,y) \wedge fol(\psi, y) $
\item $fol(\psi_1 U \psi_2, x) = \exists y.x\leq y\leq last \wedge fol(\psi_2, y) \wedge \forall z.x\leq z < y \rightarrow fol(\psi_1,z)$
\end{itemize}

In the above, the notation \emph{succ} denotes that $y$ is the successor of $x$. The following theorem guarantees a finite trace $\rho$ is a good prefix of the Co-Safety \LTL formula $\psi$ iff the corresponding interpretation $\mathcal{I}$ of $\rho$ models $fol(\psi, 0)$.

\begin{theorem}\label{thm:fol}
Given a Co-Safety \LTL formula $\psi$, a finite trace $\rho$ and the corresponding interpretation $\mathcal{I}$ of $\rho$, $\rho$ is a good prefix for $\psi$ iff $\mathcal{I}\models fol(\psi, 0)$.
\end{theorem}

\begin{proof}
We prove the theorem by the induction over the structure of $\psi$.
\begin{itemize}
\item Basically, if $\psi = p$ is an atom, $\rho$ is a good prefix for $\psi$ iff $p\in\rho_0$. By the definition of $\mathcal{I}$, we have that $0\in p^{\mathcal{I}}$. As a result, $\rho$ is a good prefix for $\psi$ iff $\mathcal{I}\models fol(p, 0)$.  Moreover, if $\psi = \neg p$ where $p$ is an atom, $\rho$ is a good prefix for $\psi$ iff $p\not\in\rho_0$, and iff $0\not\in p^{\mathcal{I}}$, finally iff $\mathcal{I}\models fol(\neg p, 0)$ holds;

\item If $\psi =\psi_1 \wedge \psi_2$, $\rho$ is a good prefix for $\psi$ implies $\rho$ is a good prefix for both $\psi_1$ and $\psi_2$. By induction hypothesis, it is true that $\mathcal{I}\models fol (\psi_1,0)$ and $\mathcal{I}\models fol(\psi_2, 0)$. So $\mathcal{I}\models fol(\psi_1,0)\wedge fol(\psi_2, 0)$, i.e. $\mathcal{I}\models fol (\psi_1 \wedge \psi_2, 0)$ holds. On the other hand, since $\mathcal{I}\models fol (\psi_1 \wedge \psi_2, 0)$, $\mathcal{I}\models fol (\psi_1,0)$ and $\mathcal{I}\models fol (\psi_2, 0)$ are true. By induction hypothesis, we have that $\rho$ is a good prefix for both $\psi_1$ and $\psi_2$. Thus $\rho$ is a good prefix for $\psi_1 \wedge \psi_2$;

\item If $\psi = \psi_1 \vee \psi_2$, $\rho$ is a good prefix for $\psi$ iff $\rho$ is a good prefix for either $\psi_1$ or $\psi_2$. Without loss of generality, we assume that $\rho$ is a good prefix for $\psi_1$. By induction hypothesis, $\mathcal{I}\models fol (\psi_1,0)$ holds, thus $\mathcal{I}\models fol (\psi_1\vee\psi_2,0)$ is true. The other direction can be proved analogously. 


\item If $\psi= X\psi_1$, $\rho$ is a good prefix for $\psi$ iff suffix $\rho'=\rho_1\rho_2\ldots,\rho_{|\rho|-1}$ of $\rho$ is a good prefix for $\psi_1$. Let $\mathcal{I}'$ be the corresponding interpretation of $\rho'$, thus every atom $p \in \P$ satisfies $i\in p^{\mathcal{I}'}$ iff $(i+1)\in p^{\mathcal{I}}$. By induction hypothesis, $\mathcal{I}'\models fol (\psi_1, 0)$ holds, thus $\mathcal{I}\models fol (\psi_1, 1)$ is true. Therefore, $\mathcal{I}\models fol (X\psi_1, 0)$ holds. 



\item If $\psi = \psi_1 U \psi_2$, $\rho$ is a good prefix for $\psi$ iff there exists $i~(0\leq i \leq |\rho|-1)$ such that suffix $\rho' = \rho_i\rho_{i+1}\ldots,\rho_{|\rho|-1}$ of $\rho$ is a good prefix for $\psi_2$. And for all $j~(0\leq j< i)$, $\rho'' = \rho_j\rho_{j+1}\ldots,\rho_{i-1}$ is a good prefix for $\psi_1$. Let $\mathcal{I}'$ and $\mathcal{I}''$ be the corresponding interpretations of $\rho'$ and $\rho''$. Thus every atom $p \in \P$ satisfies that 
$k\in p^{\mathcal{I}'}$ iff $(i+k) \in p^{\mathcal{I}}$, $k\in p^{\mathcal{I}''}$ iff $(j+k) \in p^{\mathcal{I}}$. By induction hypothesis, $\mathcal{I}'\models fol (\psi_2, 0)$ and $\mathcal{I}''\models fol (\psi_1, 0)$ holds. Thus $\mathcal{I} \models \exists i. 0\leq i\leq (|\rho|-1)\cdot fol(\psi_2,i)$ and $\mathcal{I} \models \forall j.0\leq j <i\cdot fol(\psi_1,j)$ hold. Therefore, $\mathcal{I} \models fol(\psi_1 U \psi_2,0)$.


\end{itemize}
\end{proof}

MONA~\cite{KlaEtAl:Mona} is a tool that translates \emph{Weak Second-order Theory of One or Two successors} (WS1S/WS2S)~\cite{Doner70} formula to minimal DFA, represented symbolically. WS1S subsumes the \emph{First-Order Logic} (FOL) over finite traces, which allows us to adopt MONA to construct the DFA $A_{\phi}$ for Safety formula $\phi$. Taking the assumption that the DFA generated by MONA accepts exactly the same traces that satisfy $fol(\neg\phi, 0)$, which corresponds to Co-Safety \LTL formula $\neg \phi$, by Theorem~\ref{thm:fol} we can conclude that the DFA returned by MONA is $A_{\phi}$ that accepts exactly the bad prefixes for the Safety \LTL formula $\phi$. 

\begin{theorem}
Let $\phi$ be a Safety \LTL formula and $A_{\phi}$ be the DFA constructed by MONA taking $fol(\neg\phi, 0)$ as input. Finite trace $\rho$ is a bad prefix for $\phi$ iff $\rho$ is accepted by $A_{\phi}$.
\end{theorem}

Deleting all transitions toward the accepting states in $A_{\phi}$ and removing the accepting states of $A_{\phi}$ derives the safety automaton $A^s_{\phi}$. To solve the Safety \LTL synthesis problem, we reduce the problem to a deterministic safety game over this automaton. We first present the standard formulation and algorithm for solving such a game. Then, since MONA constructs the DFA symbolically, we present a symbolic version of this algorithm.

\subsection{Solving Safety Games Symbolically}

Computing a \emph{winning strategy} of the safety game over DSA solves the synthesis problem. We base our symbolic approach on the algorithm from~\cite{DegVa15} for DFA (reachability) games, which are the duals of safety games. In this section, we first describe the general algorithm, which computes the set of winning states as a fixpoint. We then show how to perform this computation symbolically using the symbolic representation of the state graph constructed by MONA. Finally, we describe how we can use boolean synthesis to extract a winning strategy from the symbolic representation of the set of winning states.

Consider a set of states $\mathcal{E}$. The \emph{pre-image} of $\mathcal{E}$ is a set $Pre(\mathcal{E}) = \{s \in S\ |\ \forall X \in 2^{\X}. \exists Y \in 2^{\Y}. \delta(s,(X,Y)) \in \mathcal{E}\}$. That is, $Pre(\E)$ is the set of states from which, regardless of the action of the environment, the controller can force the game into a state in $\E$. If the controller moves first, we swap the order of $\exists Y \in 2^{\Y}$ and $\forall X \in 2^{\X}$ to compute the pre-image.

We define $Win(A^s)$ as the greatest-fixpoint of $Win_i(A^s)$, which denotes the set of states in which the controller can remain within $i$ steps. This means that $Win(A^s)$ is the set of states in which the controller can remain indefinitely, that is, the set of winning states. The safety game is solved by computing the fixpoint as follows: 
\begin{align}
&Win_0(A^s) = S \label{eq:win0} \\
&Win_{i+1}(A^s) = Win_i(A^s) \cap Pre(Win_i(A^s)) \label{eq:winrec}
\end{align}
That is, we start with the set of all states and at each iteration remove those states from which the controller cannot force the game to remain in the current set.

For realizability checking, if $s_0 \in Win(A^s)$, then the game is realizable, otherwise the game is unrealizable. We also consider an \emph{early-termination} heuristic to speed up the realizability checking: after each computation of $Win_i(A^s)$, if $s_0 \notin Win_i(A^s)$, then return unrealizable. To generate the strategy, we define a deterministic finite transducer $\mathcal{T} = (2^{\mathcal{X}}, 2^{\mathcal{Y}}, Q, s_0, \varrho, \omega)$ based on the set $Win(A^s)$, where: $Q = Win(A^s)$ is the set of winning states;  $\varrho : Q\times 2^{\X}\rightarrow Q$ is the transition function such that $\varrho (q, X) = \delta (q, X\cup Y)$ and $Y =\omega(q, X)$; $\omega: Q \times 2^{\mathcal{X}}\rightarrow 2^{\Y}$ is the output function, where $\omega (q, X) = Y$ such that $\delta(q, X \cup Y) \in Q$.
Note that there are many possible choices for the output function $\omega$. The transducer $\mathcal{T}$ defines a winning strategy by restricting $\omega$ to return only one possible setting of $\Y$.

Following the construction in Section~\ref{sec:construction}, MONA produces a symbolic representation of the DFA $A_\phi$ which accepts all bad prefixes of the Safety \LTL formula $\phi$. Therefore, in this section we show how to derive a DSA and solve the corresponding safety game from this representation. Following~\cite{ZTLPV17}, we define a symbolic DFA as $\A = (\mathcal{X}, \mathcal{Y}, \mathcal{Z}, Z_0, \eta, f)$, where: $\X$ is a set of input variables; $\Y$ is a set of output variables; $\Z$ is a set of state variables; $Z_0 \in 2^\mathcal{Z}$ is the assignment to the state propositions corresponding to the initial state; $\eta : 2^\mathcal{X} \times 2^\mathcal{Y} \times 2^\mathcal{Z} \rightarrow 2^\mathcal{Z}$ is a boolean function mapping assignments $X$, $Y$ and $Z$ of the variables of $\mathcal{X}$, $\mathcal{Y}$ and $\mathcal{Z}$ to a new assignment $Z'$ of the variables of $\mathcal{Z}$; $f$ is a boolean formula over the propositions in $\mathcal{Z}$, such that $f$ is satisfied by an interpretation $Z$ iff $Z$ corresponds to an accepting state.








Given $\A$, the corresponding safety automaton $A^s = (2^{\mathcal{X} \cup \mathcal{Y}}, S, s_0, \delta)$ that avoids all bad prefixes accepted by $\A$ is defined by: $\X$ and $\Y$ are the same as in the definition of $\A$; $S = \{Z \in 2^\Z \mid Z \not\models f\}$; $s_0 = Z_0$; $\delta : S \times 2^{\X \cup \Y} \rightarrow S$ is the partial function such that $\delta(Z, X \cup Y) = \eta(X, Y, Z)$ if $Z \in S$, and is undefined otherwise.





\begin{lemma}
If $\A_\phi$ is a symbolic DFA that accepts exactly the bad prefixes of a Safety \LTL formula $\phi$, then $A^s_\phi$ is a deterministic safety automaton for $\phi$.
\end{lemma}

This correspondence allows us to use the symbolic representation of $\A$ to compute the solution of the safety game defined by $A^s$. To compute the set of winning states, we represent the set $Win_i(A^s)$ by a boolean formula $w_i$ in terms of the state variables $\Z$, such that an assignment $Z \in 2^\Z$ satisfies $w_i$ if and only if the state represented by $Z$ is in $Win_i(A^s)$. We define $w_0 = \neg f$ and $w_{i+1}(Z) = w_i(Z) \land \forall X . \exists Y . w_i(\eta(X, Y, Z))$, which correspond respectively to (\ref{eq:win0}) and (\ref{eq:winrec}) above. The fixpoint computation terminates once $w_{i+1} \equiv w_i$, at which point we define $w = w_i$, representing $Win(A^s)$. We can then test for realizability by checking if the assignment $Z_0$ representing the initial state satisfies $w$.

\begin{theorem}
The safety game defined by $A^s$ is realizable if and only if $Z_0 \models w$.
\end{theorem}

If $Z_0 \models w$, then we wish to construct a transducer $\T = (2^\X, 2^\Y, Q, s_0, \varrho, \omega)$ representing a winning strategy. We define $Q = \{Z \in 2^\Z \mid Z \models w\}$, $s_0 = Z_0$ and $\varrho(Z, X, Y) = \eta(X, Y, Z)$ if $\eta(X, Y, Z) \in Q$ and undefined otherwise. To construct $\omega$, we can use a boolean-synthesis procedure. Recall that the input to this procedure is a boolean formula $\varphi$, a set of input variables $I$ and a set of output variables $O$. In our case, $\varphi(Z, X, Y) = w(\eta(X, Y, Z))$, $I = \Z \cup \X$ and $O = \Y$. The result of the synthesis is a boolean function $\omega : 2^{\Z \cup \X} \rightarrow 2^\Y$. Then, from the definition of boolean synthesis it follows that if the output is chosen by $\omega$ the game remains in the set of winning states. That is, if $Z \in 2^\Z$ satisfies $w$, then for all $X \in 2^\X$, $\eta(Z, X, \omega(Z \cup X))$ also satisfies $w$.

\section{Experimental Evaluation}

\subsection{Implementation}
\subsubsection{Explicit Approach}
The main algorithm for the explicit approach consists of three steps: DSA construction, Horn formula generation and SAT solving for synthesis. We adopted SPOT~\cite{spot} as the DSA constructor since the output automata should be deterministic. Generating the Horn formula follows the rules described in Section~\ref{sec:horn_sat}. Furthermore, here we used Minisat-2.2~\cite{minisat} for SAT solving. Decoding the variables that are assigned with the truth in the assignment returned by Minisat-2.2~\cite{minisat} is able to generate the strategy if the Safety \LTL formula is realizable with respect to $\langle inputs, outputs \rangle$. 

\subsubsection{Symbolic Encoding}
We implemented the symbolic framework for Safety \LTL synthesis in the \emph{SSyft} tool, which is written in C++ and utilizes the BDD library CUDD-3.0.0~\cite{cudd}. The entire framework consists of two steps: the DSA construction and the safety game over the DSA. In the first step, the dual of the DSA, a DFA is constructed via MONA~\cite{KlaEtAl:Mona} and represented as a~\emph{Shared Multi-terminal BDD}~(ShMTBDD)~\cite{Bryant92,KlaEtAl:Mona}. From this ShMTBDD, we construct a representation of the transition relation $\eta$ by a sequence $\mathcal{B} = \langle B_0,B_1,\ldots,B_{n-1}\rangle$ of BDDs. Each $B_i$, when evaluated over an assignment of $\X \cup \Y$, outputs an assignment to a state variable $z_i \in \Z$. The boolean formula $f$ representing the accepting states of the DFA is likewise encoded as a BDD $B_f$.

To perform the fixpoint computation, we construct a sequence $\langle B_{w_0},B_{w_1},\ldots,B_{w_i}\rangle$ of BDDs, where $B_{w_i}$ is the BDD representation of the formula $w_i$. $B_{w_{i+1}}$ is constructed from $B_{w_i}$ by substituting each state variable $z_i$ with the corresponding BDD $B_i$, which can be achieved by the \emph{Compose} operation in CUDD. Moreover, CUDD  provides the operations \emph{UnivAbstract} and \emph{ExistAbstract} for universal  and existential quantifier elimination respectively. The fixpoint computation benefits from the canonicity of BDDs by checking the equivalence of $B_{w_{i+1}}$ and $B_{w_i}$. To check realizability we use the \emph{Eval} operation. Since in our construction the state variables appear at the top of the BDDs, we use the Input-First boolean-synthesis procedure introduced in~\cite{DrLu.cav16} to synthesize the winning strategy if the game is realizable.

\subsection{Experimental Methodology}
To show the efficiency of the methods proposed in this paper, we compare our tool \emph{SSyft} based on the symbolic framework and the explicit approach, named as Horn\_SAT, with extant \LTL synthesis tools Unbeast~\cite{Unbeast} and Acacia+~\cite{Acacia}. 
Both of the \LTL synthesis tools can use either SPOT~\cite{spot} or LTL2BA~\cite{ltl2ba} for the automata construction. From our preliminary evaluation, both Unbeast and Acacia+ perform better when they construct automata using LTL2BA. As a result, LTL2BA is the default \LTL-to-automata translator of Unbeast and Acacia+ in our experiments. 
All tests are ran on a platform whose operating system is 64-bit Ubuntu 16.04, with a 2.4 GHz CPU (Intel Core i7) and 8 GB of memory. 
The timeout was set to be 60 seconds (s). 

\textbf{\emph{Input Formulas}} 
Our benchmark formulas are collected from~\cite{Unbeast}, called \emph{LoadBalancer}. Since not all cases are safe, here we propose a class of \emph{Expansion Formulas} for safety-property generation. Consider an \LTL formula $\phi$ in \NNF. We use a transformation function $ef(\phi,l)$ that given $\phi$ and a parameter $l$, which represents the expansion length, returns a Safety \LTL formula. The function $ef()$ works in the following way: (1) For each subformula of the form $\phi_1U\phi_2$, expand to $\phi_2 \vee (\phi_1 \wedge X(\phi_1U\phi_2))$ for $l-1$ times; (2) Substitute the remaining $\phi_1U\phi_2$ with $\phi_2$. Note that Safety formulas are Until-free in \NNF, thus for \LTL formulas in \NNF, it is not necessary to deal with the Release operator.  The intuition of the expansion is to bound the satisfied length of $\phi_1U\phi_2$ by adding the  Next$(X)$ operator. The parameter $l$ scales to 5 in our test, for each length there are 79 instances. And 395 cases in total.

\textbf{\emph{Correctness}}
The correctness of our implementation was evaluated by comparing the results from our approaches with those from Acacia+ and Unbeast. For the solved cases, we never encountered an inconsistency.

\subsection{Results}
We evaluated the performance of \emph{SSyft} and Horn\_SAT in terms of the number of solved cases and the running time. Our experiments demonstrate that the symbolic approach we introduced here significantly improves the effectiveness of Safety \LTL synthesis. The safety game has two versions, depending on which player (environment or controller) moves first. Both our tool \emph{SSyft} and Acacia+ are able to handle these two kinds of games, while Unbeast supports only games with the environment moving first. As a result, we only consider the comparison on the environment-moving-first game. We aim to compare the results on two aspects: 1) the scalability on the expansion length; 2) the number of solved cases in the given time limit. 

Fig.~\ref{fig:enivnum} shows the number of solved cases for each expansion length (1-5)%
\footnote{We recommend viewing the figures online for better readability.}. As shown in the figure, \emph{SSyft} solves approximately twice as many cases as the other three tools. The advantage of \emph{SSyft} diminishes as the expansion length grows, because MONA cannot generate the automata for such cases. Neither of Acacia+ and Unbeast can solve these cases even in a small expansion length. Horn\_SAT performs similarly as \emph{SSyft} when $l=1$, which derives smaller DSA. The performance of Horn\_SAT decreases sharply as the size of the DSA grows, since formula generation dominates the synthesis time. 
In total, \emph{SSyft} solves a total of 339 cases, while Acacia+, Unbeast and Horn\_SAT solve 182, 132 and 159 cases, respectively.

The scatter plot for the total time comparison is shown in Fig.~\ref{fig:enivtotal}, where $+$ plots the data for \emph{SSyft} against Acacia+, $\triangle$ plots the data for \emph{SSyft} against Unbeast and $\circ$ is for Horn\_SAT. Clearly, \emph{SSyft} outperforms the other three tools. 
The results shown in Fig.~\ref{fig:enivtotal} confirm the claim that the symbolic approach is much more efficient than Acacia+ and Unbeast. In some cases, Horn\_SAT performs better than \emph{SSyft}, nevertheless in general \emph{SSyft} has a significant advantage. Thus, the evidence here indicates that both the symbolic approach and the explicit method introduced in this paper contribute to the improvement of the overall performance of Safety \LTL synthesis. 

\begin{figure}
  \centering
  \begin{minipage}{0.45\linewidth}
  \centering
  \includegraphics[width=\linewidth]{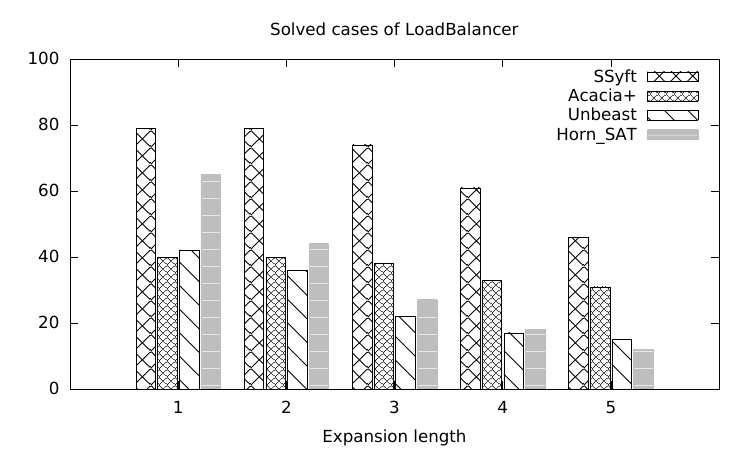}
  \caption{Comparison of SSyft against Acacia+, Unbeast and the Horn\_SAT approach on the number of solved cases as the expansion length grows}\label{fig:enivnum}
  \end{minipage}
  \hspace{0.05\linewidth}
  \begin{minipage}{0.45\linewidth}
  \centering
  \includegraphics[width=\linewidth]{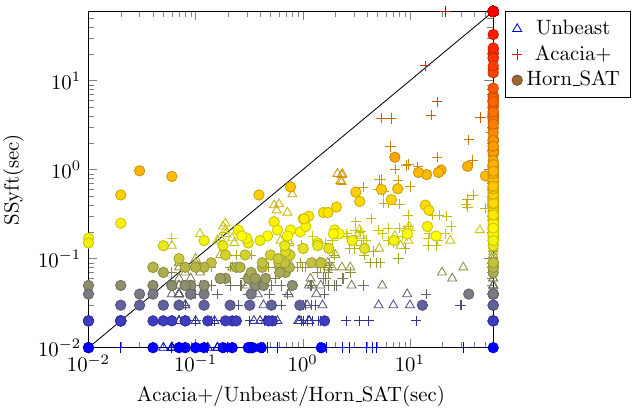}
  \caption{Comparison of SSyft against Acacia+, Unbeast and the Horn\_SAT approach on total solving time}\label{fig:enivtotal}
  \end{minipage}
\end{figure}

\section{Concluding Remarks}

We presented here a simple but efficient approach to Safety \LTL
synthesis based on the observation that a minimal DFA can be 
constructed for Co-Safety \LTL formula. Furthermore, a
deterministic safety automaton (DSA) can be generated from the
DFA, and a symbolic safety game can be solved over the DSA. A
comparison with the reduction to Horn-SAT confirms better
scalability of the symbolic approach. Further experiments show
that the new approach outperforms existing solutions for general
\LTL synthesis. Both the DSA construction and the symbolic safety
game solution contribute to the improvement.
It will be interesting to apply our approach to the
\emph{safety-first} method~\cite{SohailS09} for \LTL synthesis. 

It should be noted, however, that symbolic DSA construction
cannot avoid the worst case doubly exponential complexity: it can 
only make the synthesis simpler and more efficient in practice.
Our experiments show that the bottleneck is manifested when the
input Safety \LTL formula gets larger, and DSA construction
becomes unachievable within the reasonable time. A promising
solution may be to develop an on-the-fly method to perform the
DSA construction and solve the safety game at the same time. We 
leave this to our future work.

Beyond general \LTL-synthesis approaches, another relevant work 
is on GR(1) synthesis~\cite{gr1}. Although GR(1) synthesis aims
to handle a fragment of general \LTL as well, it is not
comparable to Safety \LTL, since GR(1) does not allow arbitrary
nesting of the Release (R) and Next (X) operators. For that 
reason, our experiments do not cover the comparison between our
approach and GR(1) synthesis. Another work related is synthesis 
of the GXW fragment~\cite{ChengHR16}. In this fragment, input 
formulas are conjuction of certain pattern formulas expressed 
using the temporal connectives $G$, $X$, and $W$. Because of the
limitation to six specific patterns, this fragment is quite less
general that the Safety \LTL fragment studied here. 

Our work is also related to the safety-synthesis track of the 
Annual Synthesis Competition~(SyntComp). While the Safety-\LTL-synthesis problem can, in principle, be reduced to safety synthesis, the reduction is quite nontrivial. Safety-synthesis
tools from SyntComp take AIGER models\footnote{\url{http://fmv.jku.at/aiger/}} as input, 
while our approach takes Safety \LTL formulas as input. A symbolic DSA can be encoded as an AIGER model by adding additional variables to encode intermediate BDD nodes. As we saw, however, the construction of symbolic DSAs is a very demanding part of Safety \LTL synthesis, with a worst-case doubly exponential complexity, so the usefulness of such a reduction is questionable.

We have shown here a new symbolic approach to Safety
\LTL synthesis, in which a more efficient automata-construction
technique is utilized. Experiments show that our new approach 
outperforms existing solutions to general \LTL synthesis, as well
as a new reduction of safety games to Horn\_SAT.

\\~\\
\noindent\textbf{Acknowledgments.} Work supported in part by NSF grants~CCF-1319459 and~IIS-1527668, NSF Expeditions in Computing project~``ExCAPE: Expeditions in Computer Augmented Program Engineering'', NSFC Projects No.~61572197 and No.~61632005, MOST NKTSP Project~2015BAG19B02, and by the Brazilian agency CNPq through the Ci\^{e}ncia Sem Fronteiras program.

\bibliographystyle{splncs03}
\bibliography{hvc}
\newpage
\appendix
\section{Appendix}
We also compare the approaches in terms of the impact of the result of the formula (realizable/unrealizable). As shown in Figure~\ref{fig:enivrea} and Figure~\ref{fig:enivunrea}, we separate the realizable and unrealizable results to explore how our new approaches perform on each category. The results show that \emph{SSyft} takes an obvious advantage on solving unrealizable cases. This happens because \emph{SSyft} benefits from the \emph{early-termination} realizability-checking heuristic; that is, the checking of whether the initial state $s_0$ is in $Win_i(\S)$ is invoked after each iteration $i\geq 0$ of computing the set of winning states, i.e. $Win_i(\S)$. Note that $s_0$ is in $Win_0(\S)$ initially, since $s_0$ is assumed to be a winning state. If $s_0$ is not in $Win_i(\S)$ for some $i$, $s_0$ is no longer a winning state such that the realizability checking indeed returns unrealizable. The explicit approach, Horn\_SAT, performs similarly as Acacia+ on realizable cases, while on unrealizable cases, Acacia+ has advantageous on Horn\_SAT. In general, Horn\_SAT performs better than Unbeast, although Unbeast has clear advantage on realizable cases.
 
\begin{figure}
  \centering
  \includegraphics[width=2.5in]{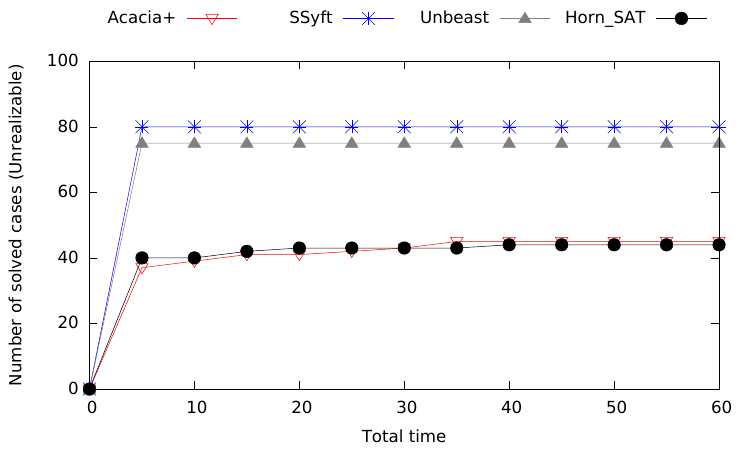}
  \caption{Comparison of SSyft against Acacia+, Unbeast and the Horn-SAT on the number of solved cases in limited time for realizable cases}\label{fig:enivrea}
\end{figure}

\begin{figure}
  \centering
  \includegraphics[width=2.5in]{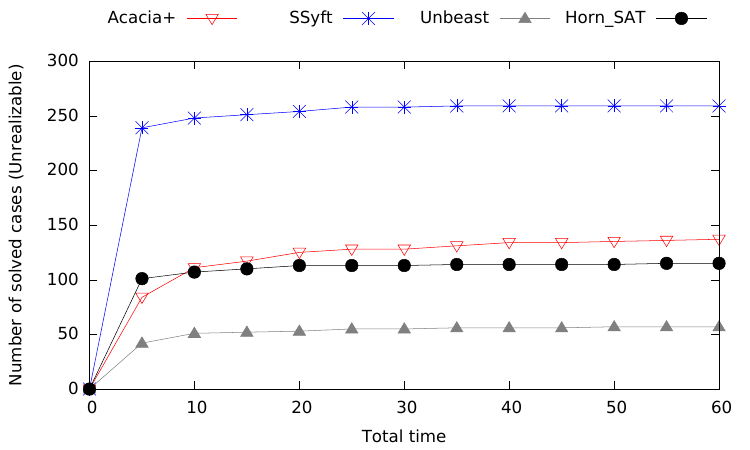}
  \caption{Comparison of SSyft against Acacia+, Unbeast and the Horn-SAT on the number of solved cases in limited time for unrealizable cases}\label{fig:enivunrea}
\end{figure}

\end{document}